\documentclass{article}

\usepackage{fixltx2e,fix-cm}

\usepackage{amssymb}
\usepackage{amsmath}
\usepackage{amsthm}[ntheorem]
\usepackage{graphicx,subcaption}
\usepackage[square]{natbib}
\usepackage{makeidx}

\usepackage{thmtools}
\usepackage{stmaryrd,amsthm,url}
\usepackage{enumerate,enumitem}
\usepackage{etoolbox,mathtools}

\newtheoremstyle{theoremdd}
  {\topsep}
  {\topsep}
  {\upshape}
  {0pt}
  {\bfseries}
  {}
  { }
  {\thmname{#1}\thmnumber{ #2}\thmnote{ (#3)}}

\theoremstyle{theoremdd}

\newtheorem{prob}{Problem}

\declaretheorem[name=Example,qed={\lower-0.3ex\hbox{~\tiny$\blacksquare$}}]{ex}
\declaretheorem[name=Definition,qed={\lower-0.3ex\hbox{~\tiny$\blacksquare$}}]{defn}
\declaretheorem[name=Remark,qed={\lower-0.3ex\hbox{~\tiny$\blacksquare$}}]{remark}
\declaretheorem[name=Theorem,qed={\lower-0.3ex\hbox{~\tiny$\blacksquare$}}]{theorem}
\declaretheorem[name=Corollary,qed={\lower-0.3ex\hbox{~\tiny$\blacksquare$}}]{cor}
\declaretheorem[name=Lemma,qed={\lower-0.3ex\hbox{~\tiny$\blacksquare$}}]{lemma}

\numberwithin{remark}{section}
\numberwithin{theorem}{section}
\numberwithin{lemma}{section}
\numberwithin{equation}{section}
\numberwithin{defn}{section}
\numberwithin{figure}{section}
\numberwithin{table}{section}
\numberwithin{ex}{section}
\numberwithin{alg}{section}
\numberwithin{prob}{section}
\numberwithin{cor}{section}

\newcommand{\bals}{\begin{align*}}

\newcommand{\eals}{\end{align*}}

\newcommand{\msppp}{\hspace{15pt}}

\newcommand{\absb}[1]{\ensuremath{\lvert#1\rvert}}

\newcommand{\ben}{\begin{enumerate}}
\newcommand{\een}{\end{enumerate}}
\newcommand{\benr}{\begin{enumerate}[label=(\roman{*}),noitemsep,leftmargin=*]}
\newcommand{\benal}{\begin{enumerate}[label=(\alph{*}),noitemsep,leftmargin=*]}
\newcommand{\benar}{\begin{enumerate}[label=(\arabic{*}),noitemsep,leftmargin=*]}

\newcommand{\benrnls}{\begin{enumerate}[label=(\roman{*}),nolistsep,leftmargin=*]}
\newcommand{\benalnls}{\begin{enumerate}[label=(\alph{*}),nolistsep,leftmargin=*]}
\newcommand{\benarnls}{\begin{enumerate}[label=(\arabic{*}),nolistsep,leftmargin=*]}

\newcommand{\bpr}{\begin{prob}}
\newcommand{\epr}{\end{prob}}

\newcommand{\bdm}{\begin{displaymath}}
\newcommand{\edm}{\end{displaymath}}

\newcommand{\beq}{\begin{equation}}
\newcommand{\eeq}{\end{equation}}

\newcommand{\bea}{\begin{eqnarray}}
\newcommand{\eea}{\end{eqnarray}}

\newcommand{\beas}{\begin{eqnarray*}}
\newcommand{\eeas}{\end{eqnarray*}}

\newcommand{\bdf}{\begin{defn}}
\newcommand{\edf}{\end{defn}}

\newcommand{\bex}{\begin{ex}}
\newcommand{\eex}{\end{ex}}

\newcommand{\bexa}{\begin{ex}}
\newcommand{\eexa}{\end{ex}}

\newcommand{\bexe}{\begin{exercise}}
\newcommand{\eexe}{\end{exercise}}

\newcommand{\bthm}{\begin{theorem}}
\newcommand{\ethm}{\end{theorem}}

\newcommand{\bmat}{\begin{bmatrix}}
\newcommand{\emat}{\end{bmatrix}}

\newcommand{\bproof}{\begin{proof}}
\newcommand{\eproof}{\end{proof}}

\newcommand{\blem}{\begin{lemma}}
\newcommand{\elem}{\end{lemma}}

\newcommand{\brem}{\begin{remark}}
\newcommand{\erem}{\end{remark}}

\newcommand{\bcor}{\begin{cor}}
\newcommand{\ecor}{\end{cor}}

\newcommand{\balg}{\begin{algorithm}}
\newcommand{\ealg}{\end{algorithm}}

\def\ng{%
  \setbox0=\hbox{-}%
  \vcenter{%
    \hrule width\wd0 height \the\fontdimen8\textfont3%
  }%
}

\newcommand{\calb}{{\cal B}}

\newcommand{\eps}{{\epsilon}}

\newcommand{\bcb}{\begin{color}{blue}}
\newcommand{\bcr}{\begin{color}{red}}
\newcommand{\bcg}{\begin{color}{green}}
\newcommand{\ec}{\end{color}}

\newcommand{\pr}{{\prime}}

\newcommand{\htheta}{{\hat{\theta}}}

\newcommand{\eqnarrayff}[4]{
	\left\{ \begin{array}{lcl} 
			#1 & ; & #2 \\
			#3 & ; & #4 \end{array}\right.
}

\newcommand{\es}[2]{E_{#2}[#1]}

\newcommand{\bsp}{\begin{sloppypar}}
\newcommand{\esp}{\end{sloppypar}}

\begin{document}

\title{Revisiting the Lost Submarine Problem: A Decision Theoretic Approach}


\author{Anthony Almudevar \\
Department of Biostatistics and Computational Biology \\ University of Rochester}

\maketitle

\noindent\textbf{Abstract.}
This article includes a discussion of the ``lost submarine problem", following Morey \emph{et al} (2016). As the title of that paper suggests (\emph{The fallacy of placing confidence in confidence intervals}), the example is intended to illustrate the futility of relying on the confidence interval as a formal inference statement.  In the view of this author,  the misgivings expressed in Morey \emph{et al} (2016) can be resolved using a decision theoretic approach. While it is true that a variety of statistical methods lead to a variety of confidence intervals, once we precisely define their purpose, a single optimal choice emerges. Furthermore, distinct purposes lead to distinct optimal choices. Therefore, that a variety of procedures exist is an advantage rather than a liability. \\



\section{Introduction}

The following problem has been widely used in the literature to test alternative principles used in the  construction of confidence intervals, as well as the principle of the confidence interval itself (\cite{morey2016}, see also \cite{welch1939}).  

\begin{ex}[The Lost Submarine Problem]\label{ex.submarine} A submarine of length $2K$ has lost contact with its support vessel. It is assumed to be stationary. A hatch is located exactly in the center. The success of a rescue attempt depends on an accurate estimation of the hatch position. 

The crew of the support vessel notice two distinct bubbles which have emerged from the submarine. This allows the crew to confine its rescue attempt to the line intersecting the position of the bubbles, which can therefore be represented by coordinates $x_1, x_2$ in the interval $[\theta - K, \theta +K]$ on that line, where $\theta$ is the position of the hatch. Furthermore, since a bubble is equally likely to emerge from any place along the length of the submarine, $x_1, x_2$ form an independent sample from a uniform distribution  on $[\theta - K, \theta +K]$. 

How can the crew best make use of this statistical model in any rescue attempt?
\end{ex}

In one sense the answer to this question is trivial. We are able to confine any rescue attempt to the line intersecting the surface positions of the bubbles.   Furthermore, we may restrict the search to $\theta \in  [\max x_i - K , \min x_i + K]$. Then suppose the rescuers have one attempt (a constraint often imposed in this problem). The strategy is to estimate $\theta$ using some  estimator $\hat{\theta} = \hat{\theta}(x_1,x_2)$, assume this is the true location, then proceed accordingly. The closer $\hat{\theta}$ is to $\theta$, the higher the probability of success. Unless there is a specific reason to assign a nonuniform prior on $\theta$, the best choice of $\hat{\theta}$ is clear. It is the unique location equivariant estimator which minimizes risk based on any symmetric loss function, which would be $\hat{\theta} = (x_1+x_2)/2$.   We then define $V = x_2 - x_1$, which is an ancillary statistic (that is, its distribution does not depend on $\theta$). 

At this point, as the problem is presented, the crew requests of a statistician an estimate of the precision of  the estimate $\hat{\theta}$, which is conventionally presented either as a confidence interval or a Bayesian credible interval $[L, U]$. Whatever form the interval takes, it is constrained to have confidence level $P_\theta(\theta \in [L,U]) = 1 - \alpha$. In \cite{morey2016} four such intervals are presented, based on four distinct principles.  To make them comparable they are constrained to have confidence level $1 - \alpha = 0.5$. Needless to say, they have very different properties, but are derived using conventional statistical principles \cite{casella2024}. These procedures will be reviewed below, but first we will describe the distributional properties of the problem.

\section{Distributional Properties}

For convenience assume $K = 1$ in Example  \ref{ex.submarine}. We are given independent observations $X_1,X_2 \sim unif(\theta - 1, \theta+1)$. The joint density is
$$
f_X(x_1,x_2) = \frac{1}{4} I\{  x_1 \in (\theta - 1, \theta+1)\} I\{  x_2 \in (\theta - 1, \theta+1)\}.
$$
Define transformation $M = (X_1 + X_2)/2$, $V = X_2 - X_1$. The determinant of the Jacobian matrix is 1. The joint density is then
\beas
 f_{M,V}(m,v) &=&  \frac{1}{4} I\{  m - v/2 \in (\theta - 1, \theta+1)\} I\{  m+v/2  \in (\theta - 1, \theta+1)\} \\
 & = & \frac{1}{4} I\{  \theta \in (m - (1-v/2), m + (1-v/2)) \}. 
\eeas
The marginal density of $V$ is 
$$
f_V(v) = \frac{1}{4} (2 - \absb{v}) I\{ v \in (-2,2) \}
$$
and the conditional density of $M \mid \{ V = v \} $ is therefore
$$
f_{M \mid V}(m \mid v) =  \frac{1}{ (2 - \absb{v}) } I\{  \theta \in (m - (1-v/2), m + (1-v/2)) \},  
$$
equivalently,   $M \mid \{ V = v \}  \sim unif(\theta  - (1-v/2),  \theta + (1-v/2))$. 

We then note that $\hat{\theta} = M$ is a minimum risk estimator (MRE) of $\theta$ under any symmetric convex loss $L(\hat{\theta}  - \theta)$, and reasonably represents a best single estimate of $\theta$ \cite{almudevar2021}.  The remaining information is contained in $V$, which is ancillary, and serves as an index of the precision of $\hat{\theta}$ through the conditional density  $f_{M \mid V}$. We necessarily have $(1-\absb{V}/2) \geq 0$, and we also know that  $\theta \in (\hat{\theta} - (1-\absb{V}/2), \hat{\theta} + (1-\absb{V}/2))$, and that  $\hat{\theta}$ is uniformly distributed within that interval.  Any statistical inference would follow from this. 

Obviously, $(M,V)$ is a 1-1 transformation of the complete observation, and so will be completely informative, but the preceding comments apply to any sample size $n$ (the number of observed bubble streams). The minimal sufficient statistic in this case is $(X_{(1)}, X_{(n)})$ and the MRE estimate $\hat{\theta} = (X_{(1)} + X_{(n)})/2$ and ancillary statistic  $V = X_{(n)} - X_{(1)}$ would play exactly the same roles \cite{almudevar2021}.  
 
\section{A Compendium of Confidence Intervals}\label{sec.submarine.ci.compendium}

If the objective is then to construct a level $1-\alpha$ confidence interval $(L,U)$ for $\theta$, there are a variety of approaches which can be taken. The procedures are summarized in Table \ref{table.submarine.ci.compendium}.

\subsection{Sampling Distribution (SD) of the Mean}

The distribution of $\hat{\theta} - \theta$ does not depend on $\theta$, and is symmetric triangular on the interval $[-1,1]$. The $q$ critical value $q \leq 0.5$ of $z = \hat{\theta} - \theta$  is $z_q = 1 - \sqrt{2q}$. Therefore, a level $1 - \alpha$ confidence interval is given by 
\beq
\hat{\theta} \pm 1 - \sqrt{\alpha}.
\eeq

\subsection{Nonparametric (NP) Procedure}

If $x_1 - \theta$ and $x_2 - \theta$ have opposite sign, then $\theta \in  [\min x_i, \max_i]$. This event clearly has probability 1/2. It is easily shown that  this confidence interval may be expressed
\beq
\hat{\theta} \pm \absb{V}/2,
\eeq
which is therefore a level $1 - \alpha = 1/2$ confidence interval.   Suppose we define a class of confidence intervals
\beq
\hat{\theta} \pm k \absb{V}.
\eeq
We may partition the event $A = \{ \theta \notin [\hat{\theta} - k\absb{V}, \hat{\theta} + k \absb{V}]$ into 
\beas
A_1 &=& \{  \hat{\theta} + k \absb{V} < \theta \} \cap \{ x_2 > x_1 \} \\
A_2 &=& \{  \hat{\theta} + k \absb{V} < \theta \} \cap \{ x_2 \leq  x_1 \} \\
A_3 &=& \{  \hat{\theta} - k \absb{V} > \theta \} \cap \{ x_2 > x_1 \} \\
A_4 &=& \{  \hat{\theta} - k \absb{V} > \theta \} \cap \{ x_2 \leq x_1 \}.
\eeas
It is easily verified that all have equal probability. Furthermore, a geometric argument shows that as a subset of the support of $(x_1,x_2)$, $A_1$ is a triangle of base $B = 1/(1/2 + k)$ and height $H = 1$. It follows that for a confidence level $1-\alpha$ we set
$$
k = k_\alpha = \frac{1}{2} \left( \frac{1-\alpha}{\alpha} \right).
$$
  
\subsection{Uniformly Most Powerful (UMP) Test}

The most powerful test (via the Neyman-Pearson Lemma) of $H_o : \theta  = \theta_0$ against   $H_a: \theta  = \theta_1$, $\theta_1 > \theta_0$, accepts $H_0$ for $\min x_i \leq \theta_0 + c$, $\max x_i \leq \theta_0 + 1$. This gives a lower bound of the form
$$
\theta \geq \hat{\theta} - \min\left( \absb{V}/2 + c,  (1-\absb{V}/2) \right).
$$
The upperbound has essentially the same form, giving confidence interval
$$
\hat{\theta} \pm \min\left( \absb{V}/2 + k_\alpha,  (1-\absb{V}/2) \right). 
$$
It may be shown that confidence $1-\alpha$ is attained with $k_\alpha = 1 - \sqrt{2\alpha}$. 

\subsection{Bayesian Credible (BC) Interval}

The Bayesian credible interval is the central $(1-\alpha)$ portion of the support of  the likelihood,  
\beq
\hat{\theta} \pm (1-\alpha)(1 - \absb{V}/2).
\eeq
which has confidence level $(1-\alpha)$. 

\begin{table}
\centering
\caption{Summary of confidence intervals of Section \ref{sec.submarine.ci.compendium}: 
(SD) Sampling Distribution; (NP)  Nonparametric Procedure; (UMP) UMP Test; (BC) Bayesian Credible Interval.
Critical values $k_\alpha$ yield a confidence level of $1-\alpha$.}\label{table.submarine.ci.compendium}
\begin{tabular}{l|l|l}\hline
Derivation & & \\ 
 Principle & $\hat{\theta} \pm b(V)$ & $1 - \alpha$ critical value $k_\alpha$ \\ \hline 
(SD) & $\hat{\theta} \pm k_\alpha$ & $k_\alpha = 1 - \sqrt{\alpha}$ \\
(NP) & $\hat{\theta} \pm k_\alpha\absb{V}$ & $k_\alpha = (1-\alpha)/(2\alpha)$ \\
(UMP) & $\hat{\theta} \pm \min\left( \absb{V}/2 + k_\alpha,  (1-\absb{V}/2) \right)$ & $k_\alpha = 1 - \sqrt{2\alpha}$ \\
(BC) & $\hat{\theta} \pm k_\alpha (1 - \absb{V}/2)$ & $k_\alpha = 1 - \alpha$
\end{tabular}
\end{table}

\section{What Properties Must the Confidence Interval Possess?}\label{sec.ci.submarine.properties}

Principles of statistical inference occasionally force a unique choice of procedure. Sometimes this class consists of a single procedure, and when it doesn't, an optimization method  may be available to select a procedure which is best according to a well defined and intuitive criterion. 

However, sometimes there is a real choice. Suppose $X_1, X_2$ are $iid$ Bernoulli random variables with mean $\theta$. The complete sufficient statistic for $\theta$ is $T = X_1 + X_2$,  so according to the sufficiency principle 
of inference any estimator of that estimand should be a function of $T$.  Consider three estimators $\htheta_1 = (X_1+X_2)/2$,  $\htheta_2 = \htheta_1/2 + 1/4$, $\htheta_3 = (3X_1+X_2)/4$. The sufficiency principle forces us to eliminate $\htheta_3$, but is satisfied by $\htheta_1$ and $\htheta_2$. The conventional choice is probably $\htheta_1$, which is unbiased. On the other hand, if we  consider mean squared risk $MSE(\theta; \htheta) = E_\theta[(\htheta - \theta)^2]$, we find that $\htheta_2$ possesses both smaller average risk over a uniform prior on $\theta$, and smaller minimax risk \cite{almudevar2021}.
 
So, how does this extend to a confidence interval? Note that a confidence interval is not a characterization of a specific  estimator.  They are often derived from the standard error of an estimator, but this is simply a special case, since the conventional definition makes no reference to any estimation problem. It is a statistic on its own, and is therefore associated with a decision problem which may be related to, but is not equivalent to, the estimation problem.   

So, as we did with estimation and hypothesis testing, we can first  describe properties that a confidence interval $(L,U)$ should possess.
  
\subsection{Location Equivariance}

The endpoints $L,U$ should be equivariant. This means the inference procedure should be invariant to any transformation $\theta + c$. This simply says that the definition of the origin of the axis on which $\theta$ is located is irrelevant. Equivalently, the prior density of $\theta$ is uniform.  It can be proven that we may select any equivariant statistic $\delta_0$, so that any other equivariant statistic may be written $\delta = \delta_0 + w$, where $w$ is a location invariant statistic. Here, the obvious choice is $\delta_0 = \hat{\theta}$.  Any invariant statistic is a function of $V$, which is a maximal invariant \cite{almudevar2021}. Therefore, we can write
$$
L = \hat{\theta} - b_L(V), \,\,\, U = \hat{\theta} + b_U(V).
$$
Note that for $n > 2$, the sufficiency principle would require that any statistic be a function only of  $(X_{(1)}, X_{(n)})$, and we could similarly argue that any invariant statistic must be a function of $V = X_{(n)} - X_{(1)}$.

\subsection{Symmetry} 

The search for optimal confidence intervals can be confined to symmetric ones, that is $\hat{\theta} \pm b(V)$. Consider the following theorem.   

\begin{theorem}\label{thm.min.ci.submarine}
Suppose $U \sim unif[\theta - L, \theta + L]$. For any $1 - \alpha \in [0,1]$, there exists a confidence interval $(U - c_1, U + c_2)$ with exact confidence level $1 - \alpha$. Among all such confidence intervals the minimum  width is $L(1 - \alpha)$ and is always attained by the symmetric confidence interval $U \pm  L(1 - \alpha)/2$. 
\end{theorem}

\begin{proof}
Suppose the confidence interval $[U - c_1, U + c_2]$ has exact confidence level 
\begin{align*}
&P( U \in [\theta - c_2,  \theta + c_1] \cap [  \theta - L,  \theta + L])& \\  
&\msppp = P( U \in [\theta - \min(c_2,L),   \theta + \min(c_1,L) ] ) \\
&\msppp =  \frac{\min(c_1,L) + \min(c_2,L)}{2L} \\
&\msppp\leq  \frac{c_1 + c_2}{2L}.  
\end{align*}
Thus, the width is bounded below by $c_1 + c_2 \geq (1-\alpha) 2L$. It is easily verified that the symmetric  confidence interval $U \pm L(1-\alpha)$ has a confidence level $1-\alpha$ and attains the minimum possible width. 
\end{proof}

\subsection{Admissibility}

 With probability 1 we have $\theta \in (\hat{\theta} - (1-\absb{V}/2), \hat{\theta} + (1-\absb{V}/2))$. Therefore, a confidence interval $\hat{\theta} \pm b(V)$ should always satisfy  $b(V) \leq  (1-\absb{V}/2)$. A confidence interval which does not satisfy this inequality can be regarded as inadmissible. 
 
\begin{theorem}\label{thm.ci.submarine.admissible}
A confidence interval for which $P_\theta\{b(V) > (1-\absb{V}/2)\} > 0$ may be replaced   with a confidence interval $\hat{\theta} \pm b^\pr(V)$, with the same confidence level,  and for which 
\beas
b^\pr(v) &\leq& b(v), \\
P_\theta\{ b(V) \geq b^\pr(V) \} &>& 0. 
\eeas  
\end{theorem}

\begin{proof} 
Set $b^\pr(V) = \min(b(V), (1-\absb{V}/2))$. The remaining argument follows Theorem  \ref{thm.min.ci.submarine}.
\end{proof}

\subsection{Comparing the Procedures}

If we examine the four candidate procedures (Table \ref{table.submarine.ci.compendium}), it is clear that two procedures (SD) and (NP) are inadmissible. However, since the admissibility bound  $(1-\absb{V}/2)$ is observable, it seems reasonable to replace any confidence interval $\hat{\theta} \pm b(V)$ with $\hat{\theta} \pm \min(b(V),(1-\absb{V}/2))$. In fact, by Theorem \ref{thm.ci.submarine.admissible} the critical values $k_\alpha$ listed in  Table \ref{table.submarine.ci.compendium} will still define an exact confidence level $1-\alpha$ for the modified procedure. 

Figure \ref{fig1-submarine} plots the bound  functions $b(v)$ for the four procedures, for confidence levels $1 - \alpha = $ 50\%, 75\%. Any portion of $b(v)$ exceeding the admissible bound $(1 - \absb{v})/2$ is colored gray. The UMP procedure coincides with the admissible bound for all large enough $\absb{v}$.  

Denote the bounds for the respective procedures $b_{SD}(v)$, $b_{NP}(v)$, $b_{UMP}(v)$, $b_{BC}(v)$. Clearly, when we consider inadmissible procedures, there is considerable variety of behavior. For example  $b_{NP}(v)$ increases with $\absb{v}$ while $b_{BC}(v)$ is decreases with $\absb{v}$.  On the other hand $b_{UMP}(v)$  possess a maximum in the interior $\absb{v} \in (0,2)$, while $b_{SD}(v)$ is constant. 

Application of the admissibility bound seems to impose at least some uniformity of behavior. For all procedures, $b(2) = 0$, indicating that the error of the estimate $\hat{\theta}$ approaches 0 as $\absb{v}$ approaches 2. However, now $b_{NP}(v)$ and $b_{UMP}(v)$ both possess a maximum in $\absb{v} \in (0,2)$, while $b_{SD}(v)$ remains constant for all small enough $\absb{v}$. In contrast, $b_{BC}(v)$ is the only bound function which remains strictly monotone.   Thus, a significant variety of choices among procedures remains.

\begin{figure}
\centering
\includegraphics[width=4.6in, height=2.8in, viewport=20 140 450 370, clip]{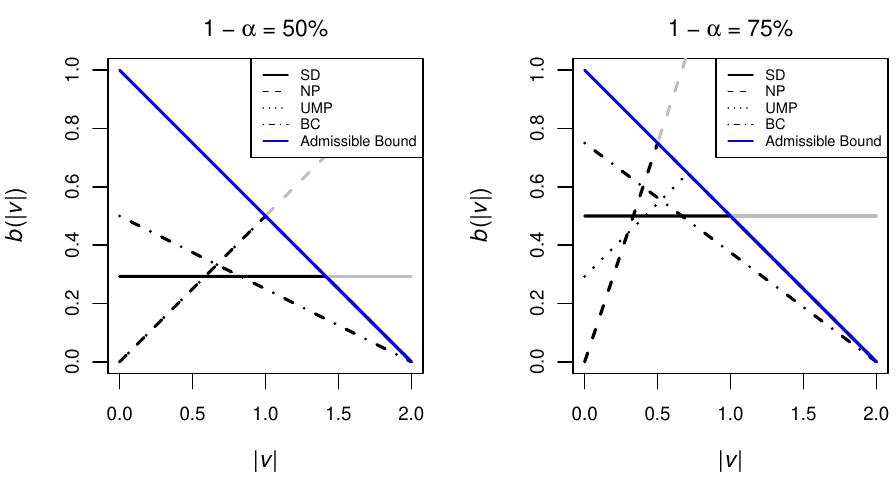}
\caption{\small Plots of confidence bounds $b(v)$ for procedures given in Table \ref{table.submarine.ci.compendium}, for confidence levels $1 - \alpha = $ 50\%, 75\%.  Gray lines indicate where $b(v)$ exceeds admissible bound $(1 - \absb{v})/2$. Note that for     $1 - \alpha = $ 50\%, the NP and UMP procedures coincide (Section \ref{sec.ci.submarine.properties}).}\label{fig1-submarine}
\end{figure}

\section{A Decision Theoretic Approach to the Selection of Confidence Intervals}\label{sec.decision.theoretic.ci}

At this point, we have enumerated four procedures satisfying the principles given in Section \ref{sec.submarine.ci.compendium}, once the admissibility bound is applied. However, the behavior of these procedures still varies considerably, so we still need to determine on which basis a choice of procedure can be made. 

Suppose we adopt a more explicit decision theoretic framework. We might start by asking how any inference will be used by the rescue crew of Example \ref{ex.submarine}. So far, we have an estimate $\hat{\theta}$ of $\theta$ which is in some sense optimal. We also have  a framework with which to assess the accuracy of the estimate, in the form of the conditional density $\hat{\theta} \mid V$.  If the crew makes a single rescue attempt, accepting $\hat{\theta}$ as the best estimate, then presumably the rescue attempt will succeed if $\hat{\theta}$ is within some error tolerence $\eps$ of $\theta$. Thus, $\hat{\theta}$ should be used because it minimizes the expected distance from $\theta$. But the action of the crew is in no way affected by the choice of confidence interval. More generally, any assessment of the precision of  $\hat{\theta}$ plays no role in any decision.

\subsection{Decision Theoretic Formulation}

Suppose we can more precisely describe a property of a confidence interval that can be optimized. This can be done in terms of a loss function $L_\theta(L,U)$ where $(L,U)$ is a confidence interval for $\theta$. The general approach will be as follows.  We accept the conclusion of Theorem \ref{thm.min.ci.submarine}, and consider only confidence intervals of the form $\hat{\theta} \pm b(v)$, taking as given that among two procedures with the came confidence level, the one with the smaller width is preferable. 

Define the conditional confidence level 
\beq
1 - \alpha(v) = P_\theta( \theta \in (L,U) \mid V = v).  \label{eq.alpha.v}
\eeq
If $(L,U)$ has confidence level $1-\alpha$, then
\beq
\int_{v \in (-2,2)}  (1 - \alpha(v)) f_V(v) dv = 1 - \alpha.  \label{eq.alpha.v.int}
\eeq
If $L_\theta$ is location invariant, then risk   
$$
R(\theta) = \es{L_\theta(L,U)}{\theta} 
$$ 
is independent of $\theta$. Thus, we can proceed by minimizing $\es{L_\theta(L,U) \mid V=v}{\theta}$ for each $v$. 
This is easily done, since conditional on $\{V=v\}$, by Theorem  \ref{thm.min.ci.submarine} the smallest interval with coverage $1 - \alpha(v)$ is $\hat{\theta} \pm (1 - \alpha(v))(1 - \absb{v}/2)$. Assuming loss is nondecreasing in confidence interval size, this means we can construct a minimum risk procedure by setting
$$
b(v) = (1 - \alpha(v))(1 - \absb{v}/2).
$$
The confidence level constraint is then given by applying the law of total probability 
$$
R(\theta) = \es{L_\theta(L,U)}{\theta} 
 = \es{\es{L_\theta(L,U) \mid V}{\theta}}{\theta},
$$
which is equivalent to 
\beas
1 - \alpha &=& \int_{v \in (-2,2)}  (1 - \alpha(v)) f_V(v) dv \\
&=& \int_{v \in (-2,2)}  \frac{b(v)}{(1-\absb{v}/2)} \frac{1}{4} (2 - \absb{v}) dv \\
&=& \frac{1}{2}\int_{v \in (-2,2)}  b(v) dv.
\eeas
The final step is to select the appropriate loss function.

\subsection{Minimum Search Effort}\label{sec.minimum.search.effort}

Suppose the strategy of the rescue crew is to search exhaustively for $\theta$ within some interval. If a level $1-\alpha$ confidence interval $\hat{\theta} \pm b(V)$ is selected for this purpose, then the rescue will succeed with probability $1-\alpha$. However, it may be important to limit the time taken by the search effort, which is assumed to be proportional to the length of the confidence interval.  Therefore, the confidence interval which minimizes the expected length for  a fixed confidence level is selected. This is deemed reasonable, since the expected length is minimized while holding the probability of a successful rescue fixed. 

This strategy is attained by adopting loss function $L_\theta = 2b(v)$. The following theorem gives the minimum risk choice for $b(v)$. 

\begin{theorem}\label{thm.submarine.loss.1}
Define $\calb$ to be the class of real valued functions on domain $[-2,2]$. Consider the problem of minimizing functional $\Gamma(b)$, $ b \in \calb$:
$$
\Gamma(b) =  \int_{v \in (-2,2)} b(v) f_V(v) dv 
$$
subject to
\bea
b(v) &\geq& 0, \label{thm.sub.2} \nonumber \\
b(v) &\leq& (1 - \absb{v}/2), \label{thm.sub.3} \nonumber  \\
\frac{1}{2} \int_{v \in (-2,2)} b(v) dv & = &1-\alpha.\label{thm.sub.4} 
\eea
Then $\Gamma(b)$ is minimized by 
 \beq
b(v) = \eqnarrayff{(1- \absb{v}/2)}{\absb{v} \geq k_\alpha}{0}{\absb{v} < k_\alpha} \label{eq.min.search.effort}
\eeq
 where $k_\alpha = 2(1-\sqrt{1-\alpha})$. 
\end{theorem}

\begin{proof}
Since $\Gamma(b)$ is an integral we may confine attention to $b \in \calb$ which possess one-sided derivatives everywhere.  Define the family of open intervals $I_v = (0, (1 - \absb{v}/2))$. Suppose for $0 \leq v_1 < v_2$ we have $b(v_1) \in I_{v_1}$, $b(v_2) \in  I_{v_2}$, Suppose $b(v_1) + b(v_2) = q$. Then, since $f_V(v)$ is strictly decreasing in $\absb{v}$, $\Gamma(b)$ can be strictly improved by increasing  $b(v_2)$ and decreasing $b(v_1)$ by equal amounts, leaving the constraints satisfied.  Therefore $b$ cannot be optimal unless $b(v_1) = (1 - \absb{v}/2)$ or $b(v_2) = 0$.  This is only possible if $b(v)$ equals \eqref{eq.min.search.effort}, with $k_\alpha$ selected to satisfy the constraints. 
\end{proof}

The solution offered by Theorem \ref{thm.submarine.loss.1}  (Figure \ref{fig2-submarine} (a)) is rather surprising, and is quite unlike any confidence interval obtained by conventional methods. It is, however, the optimal implementation of a strategy which is precisely stated and arguably reasonable. Thus, the merit of this solution is no longer a mathematical question. 

It can, however, be explained in terms of economic cost. The larger $\absb{v}$ is, the smaller is the region in which $\theta$ can be located with certainty. The meaning of Equations \eqref{eq.alpha.v} and  \eqref{eq.alpha.v.int} is that the probability of success $1-\alpha$ can be allocated in any way to a class of confidence intervals indexed by $v$. The solution offered by  Theorem \ref{thm.submarine.loss.1} simply allocates all of that probability to the shortest confidence intervals.  

If the minimum search effort solution of Theorem \ref{thm.submarine.loss.1}  is unsatisfactory, we can try taking a more refined approach to defining the decision problem, which we do in the next section.

\begin{figure}
\centering
\includegraphics[width=4.6in, height=2.8in, viewport=20 140 450 380, clip]{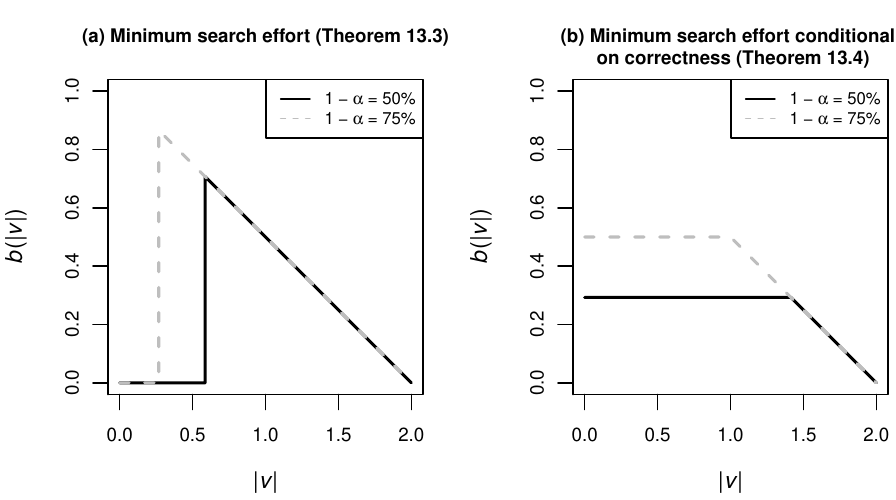}
\caption{\small Plots of confidence bounds $b(v)$ for procedures given in Theorem \ref{thm.submarine.loss.1} (Plot (a)) and Theorem \ref{thm.submarine.loss.2} (Plot (b)), for confidence levels $1 - \alpha = $ 50\%, 75\%.}\label{fig2-submarine}
\end{figure}

\subsection{Minimum Search Effort Condtional on Correctness of the Confidence Interval}

The dependence of this problem on the decision theoretic point of view can be illustrated further by considering a  modification of the rescuers' objectives. It could, for example, be argued that if the motive for a shorter confidence interval is the safety of those being rescued, then its length is only important if it contains $\theta$. In other words, we wish to minimize the expected length of the confidence interval conditional on its being correct. 

This is, in fact, a well defined problem. In this case the loss is, for some constant $K$, 
$$
L_\theta(L,U) = (U-L)I\{\theta \in (L,U)\} + K I\{\theta \notin (L,U)\}.
$$   
This loss conforms to indifference to the length of the  confidence interval when it does not contain $\theta$. The risk is then
$$
R_\theta = \es{(U-L)I\{\theta \in (L,U)\}}{\theta} + K \alpha.
$$   
But the problem is constrained by forcing $\alpha$ to be constant. This means the optimization problem is equivalent for all values of $K$, which may therefore be set to zero. We therefore define the loss function
$$
L_\theta(L,U) = (U-L)I\{\theta \in (L,U)\}.
$$   
Suppose the confidence interval is given by $\hat{\theta} \pm b(v)$. Then 
\beas
\es{L_\theta \mid V = v}{\theta} & = & \es{2b(v) I\{\theta \in (\hat{\theta} - b(v), \hat{\theta} + b(v))\} \mid V = v}{\theta} \\
&=& 2b(v) \es{I\{ \theta \in (\hat{\theta} - b(v), \hat{\theta} + b(v)) \} \mid V = v}{\theta} \\
& = & 2b(v) \frac{2b(v)}{2-\absb{v}} \\
& = & \frac{4b(v)^2}{2-\absb{v}}. 
\eeas
Similar to the previous procedure, we set 
 \beas
R_\theta &=& \int_{v \in (-2,2)}  \frac{4b(v)^2}{2-\absb{v}} f_V(v) dv \\
& = &   \int_{v \in (-2,2)}  b(v)^2 dv.
\eeas
The constraints are otherwise the same.

\begin{theorem}\label{thm.submarine.loss.2}
Define $\calb$ to be the class of real valued functions on domain $[-2,2]$. Consider the problem of minimizing functional $\Gamma(b)$, $ b \in \calb$:
$$
\Gamma(b) =  \int_{v \in (-2,2)} b(v)^2 dv 
$$
subject to
\bea
b(v) &\geq& 0, \label{thm.sub.2} \\
b(v) &\leq& (1 - \absb{v}/2), \label{thm.sub.3}  \\
\frac{1}{2} \int_{v \in (-2,2)} b(v) dv & = &1-\alpha.\label{thm.sub.4} 
\eea
Then $\Gamma(b)$ is minimized by  
\beq
b(v) = \eqnarrayff{(1- \absb{v}/2)}{\absb{v} \geq \sqrt{4\alpha}}{k_\alpha}{\absb{v} < \sqrt{4\alpha}} \label{eq.min.odds}
\eeq
where $k_\alpha =   1 - \sqrt{\alpha}$.
\end{theorem}

\begin{proof}
Since $\Gamma(b)$ is an integral we may confine attention to $b \in \calb$ which possess one-sided derivatives everywhere.  First, $b(v)$ must be continuous. To verify this, suppose for some $v^\pr$, we have $\lim_{\eps\downarrow 0} b(v^\pr+\eps) = b(v^\pr_+)  \neq b(v^\pr)$. Then  
$$
\liminf_{\eps\downarrow 0} b(v^\pr+\eps)^2 + b(v^\pr)^2 -  2 \left( \frac{b(v^\pr_+)+ b(v^\pr)}{2} \right)^2 > 0.
$$
This implies that $b(v)$ can be strictly improved by setting $b(v) = (b(v^\pr_+)+ b(v^\pr))/2$ in a small enough neighborhood of $v^\pr$. 

Then define the family of open intervals $I_v = (0, (1 - \absb{v}/2))$.  Suppose for $0 \leq v_1 < v_2$ we have $b(v_1) \in I_{v_1}$, $b(v_2) \in  I_{v_2}$. Suppose $b(v_1) + b(v_2) = q$. Then $\Gamma(b)$ cannot be minimized unless $b(v_1) = b(v_2)$. Since $b(v)$ is continuous, this implies that $\Gamma(b)$ is minimized by \eqref{eq.min.odds}. 
\end{proof}

For this loss, the minimum risk procedure is therefore the (SD) procedure of Table \ref{table.submarine.ci.compendium}, after applying the admissibility bound (Figure \ref{fig2-submarine} (b)). That this is quite different from the solution offered by Theorem \ref{thm.submarine.loss.1} illustrates an important point, that no one single method of inference can be regarded as uniformly best. Once we have identified which procedures are admissible, we cannot identify a best choice without the formulation of a decision problem.  As we see here, this aspect of the inference problem, that is, the careful selection of the loss function, can be the most important.


\end{document}